\documentclass[a4paper]{amsart}

\date{Revised on May 20, 2012}

\newtheorem{lemma}{Lemma}
\newtheorem{theorem}{Theorem}

\newtheorem{corollary}{Corollary}

\def\cz{\mathbb{C}} 
\def\rz{\mathbb{R}} 
\def\nz{\mathbb{N}} 
\def\gz{\mathbb{Z}} 

\newcommand{\cE}{\mathcal{E}}

\def\gp{\mathfrak{p}}
\def\gq{\mathfrak{q}}

\def\gA{\mathfrak{A}}

\def\gH{\mathfrak{H}}
\def\gQ{\mathfrak{Q}}

\def\rd{\mathrm{d}}
\def\ri{\mathrm{i}}

\def\eh{{1\over2}}

\def\d{\mathop{\partial}\nolimits}

\def\sgn{\mathop{\mathrm{sgn}}\nolimits}

\def\f#1#2{\phi_{#1,#2}}
\def\p#1#2{\psi_{#1,#2}}
\def\v#1#2#3{v^{#1}_{#2,#3}}

\begin{document}
\title[Stability of Graphene Impurites]{Stability of Impurities with Coulomb Potential in Graphene with
  Homogeneous Magnetic Field}

\author[t. Maier]{Thomas Maier}
\author[H. Siedentop]{Heinz Siedentop}
\address{Mathematisches Institut\\
 Ludwig-Maximilians-Universit\"at M\"unchen\\
 Theresienstra\ss e 39\\ 80333 M\"unchen\\ Germany}
 \email{maier@math.lmu.de  \textit{and} h.s@lmu.de}

\maketitle
\begin{center}
  \textit{Dedicated to Elliott H. Lieb on the occasion of his 80th birthday}
\end{center}
\begin{abstract}
  Given a 2-dimensional no-pair Weyl operator $W_Z$ with a point
  nucleus of charge $Z$, we show that a homogeneous magnetic field
  does not lower the critical charge beyond which it collapses.
 \end{abstract}

\section{Introduction\label{ein}}

Perfect graphene is modeled in continuum one-particle approximation by
a two-dimensional free Weyl operator (massless Dirac
operator). Non-perfect graphene has additional potentials; a
particular case of importance is the presence of an impurity of
Coulomb type (see the review of Castro Neto et
al. \cite{CastroNetoetal2009}). As opposed to non-relativistic
mechanics, in relativistic mechanics both kinetic energy and the
Coulomb potential energy have the same linear scaling for large
momenta which implies the existence of a critical coupling
constant. This explains the interest in the subject in the physics
literature, see, e.g., Pereira et al. \cite{Pereiraetal2007} and
Shyvtov et al. \cite{Shytovetal2007}. The critical coupling constant
as occurring in these papers can be mathematically thought of as that
coupling constant were a natural definition of self-adjointness ceases
to exist. In addition to the electric impurity potential it is often
also important to study the systems with an additional homogeneous
magnetic field perpendicular to the graphene sheet. Of course, the
question arises to what extend the presence of the magnetic field
changes the critical coupling constant.

If one is interested in multi-particle effects it is essential to have
a well defined multi-particle Hamiltonian (see \cite{Eggeretal2010}
and the references therein). Because of the Weyl operator's
unboundedness from above and below, a naive addition of the
one-particle operators acting on the various particles plus their
interactions -- as would be natural in non-relativistic quantum
mechanics -- does not give meaningful Hamiltonians (Brown and
Ravenhall \cite{BrownRavenhall1951}). This problem can -- on a
physical level -- be overcome by a quantum field theoretical
treatment. Approximately, one can use the no-pair Hamiltonians
initially introduced by Brown and Ravenhall \cite{BrownRavenhall1951}
and further developed by Sucher \cite{Sucher1980}. Because a
non-perturbative analytic treatment of quantum electrodynamics is not
available, we will concentrate on the second alternative.

A description of -- one-particle -- no-pair operators in an nutshell
is as follows: the state space on which the no-pair Hamiltonians are
defined depend on a Dirac sea in a similar way as the Fock
representation of the electron-positron field depends on the initial
splitting of the Hilbert space into electron and positron space (see
Thaller \cite[Section 10.1.1]{Thaller1992}). The Dirac sea is defined
through an orthogonal projection $(1-\Lambda)$ in the state space of
the Weyl operator $W_Z$, i.e., the Hilbert space
$L^2(\cz^2,\rz^2)$. The projection $\Lambda$ is assumed fixed. The
physically allowed states of Dirac particles will be those which are
orthogonal to the sea, i.e., they are eigenstates of
$\Lambda$. Metaphorically speaking the physical states are the vapor
above the Dirac sea. The no-pair approximation will then be the Weyl
operator projected onto the states of fixed particle number $N$ -- in
our case $N=1$ -- in the vapor, i.e., $B_z:=\Lambda W_Z\Lambda$.

It is reasonable to expect that expectations of the no-pair
Hamiltonians $B_Z$ are bounded from below if $Z$ is small (close to
zero) and is unbounded from below for large $Z$. The critical coupling
$Z_c$ constant is the value of the coupling constant where this change
of behavior occurs. A priori $Z_c$ can be expected to depend on the
choice of the Dirac sea. A particular simple choice is to take the
Dirac sea as the one defined by the Weyl operator with the external
homogeneous field. It is exactly this operator which we will be
interested in. Our goal is to show that $Z_c$ does not depend on the
presence of a homogeneous magnetic field.

Although the rest of the paper is on the one-particle level, our
interest in the one-particle stability stems from the multi-particle
stability: the multi-particle energy of the no-pair Hamiltonian is
bounded from below if and only if the corresponding one-particle
Hamiltonian is bounded from below (see \cite{Eggeretal2010}).

Our contribution is organized as follows: to escape the
inconclusiveness of heuristic arguments, we give a precise
mathematical formulation of the problem, collect some well known
relevant facts and state our result (Section \ref{sec:NotandForm}). To
prepare for the proof we study the partial wave analysis of underlying
energy form in Section \ref{sec:posspecpartwave}. In Section
\ref{sec:proof} we give the actual proof of our claim. The appendices
contain auxiliary material which we collect for the convenience
of the reader.

\section{Notation, Formulation of the Problem, and Main
  Result \label{sec:NotandForm}}

The Weyl operator (massless Dirac operator) of a particle of charge
$-e$ in two dimensions in a magnetic field
$\partial_1A_2-\partial_2A_1$ with vector potential $\gA:=(A_1,A_2)$
and an electric potential $\varphi$ is given by
\begin{equation}
 \label{eq:1}
 W_{\gA,\varphi}:=v\boldsymbol\sigma\cdot(\gp+\frac ec\gA)-e \varphi
\end{equation}
where $v$, $c$, and $e$ are positive constants. Depending on the
application, $v$ could be, e.g., the velocity of light or the Fermi
velocity in graphene, $c$ is the velocity of light, and
$\boldsymbol\sigma=(\sigma_1,\sigma_2)$ are the first two Pauli
matrices, i.e.,
$$
\sigma_1=
\begin{pmatrix}
 0&1\\
 1&0
\end{pmatrix},\
\sigma_2=
\begin{pmatrix}
 0&-\ri\\
 \ri&0
\end{pmatrix}.
$$
In this paper we are mainly interested in the case of a homogeneous
magnetic field of strength $B>0$ orthogonal to the $x_1$-$x_2$-plane,
i.e., $ \gA(x) = \tfrac B2(-x_2,x_1)$, and an electric field generated
by a nucleus of atomic number $Z$, i.e., $\varphi(x) = Ze/|x|$. (Note
that $B<0$ corresponds just to a reflection of the coordinates
$x\to-x$. For $B=0$ see Remark \ref{nummer} of Section
\ref{sec:NotandForm}.) This operator is to be self-adjointly realized
in $L^2(\rz^2,\cz^2)$.  Following Brown and Ravenhall
\cite{BrownRavenhall1951} -- see also Sucher
\cite{Sucher1980,Sucher1984} -- we will project these operators to the
orthogonal space of a Dirac sea.  More precisely, we are interested in
the quadratic form of $W_{\gA,\varphi}$ restricted to the positive
spectral subspace
$$\gH:=\{\psi\in L^2(\rz^2,\cz^2)\ |\  \psi=\Lambda\psi\}$$ 
with $\Lambda:= \chi_{(0,\infty)}(W_{\gA,0})$.

By dilation $W_{\gA,\varphi}$ is unitarily equivalent to $\sqrt{eB/(2c\hbar)}v
W_{e^2Z/(\hbar v)}$. Thus it suffices to study $W_{Z}:=
\boldsymbol\sigma\cdot(\gp+(-x_2,x_1)) - \frac{Z}{|x|}$ assuming that
$e=v=\hbar=1$. We use complex notation $z:= x_1+\ri x_2$, and,
correspondingly $\bar\d:= \tfrac12(\partial_1+\ri\partial_2)$ and
$\d:= \tfrac12(\partial_1-\ri\partial_2)$ and introduce $d:=
-2\ri(\d+\bar z/2)$ and $d^*:=2\ri(-\bar\partial+z/2)$. This allows us
to write more compactly
\begin{equation}
 \label{eq:3}
 W_Z =
 \begin{pmatrix}
   0& d\\
   d^*& 0
 \end{pmatrix}
 - \frac Z{|\cdot|}.
\end{equation}

We define $\gq_0$ as the linear span of the functions $\f mn$ defined
in \eqref{eq:6} of the Appendix. 
We also define $\gQ_0$ as the linear span of
the spinors $\p mn$ defined in \eqref{eq:5} of the Appendix. 
\begin{theorem}
 \label{theorem1}
 The quadratic form $(\psi,W_Z\psi)$ is positive on $\gQ_0$ and
 extends to a closed form $\cE$ on $\gQ$ which is bounded from below,
 if
 \begin{equation}
   \label{eq:4}
   Z \leq Z_c :=\left(\frac{\Gamma(\tfrac14)^4}{8\pi^2}+\frac{8\pi^2}{\Gamma(\tfrac14)^4}\right)^{-1}.
\end{equation}
For $Z>Z_c$ the form is unbounded from below.
\end{theorem}
We remark:
\begin{enumerate}
\item Physically $\cE[\psi]$ is the energy of an electron in the state
  $\psi$ on top of the Fermi sea defined by $\gH^\perp$ in the quantum
  dot defined by the homogeneous magnetic field and an interstitial
  atom with charge $Z$.
\item If $Z\leq Z_c$, then the form $\cE$ defines -- according to
  Friedrichs \cite[Satz 3]{Friedrichs1934} -- a unique positive self-adjoint
  operator whose form domain includes $\gq_0$ and extends $\Lambda
  W_z$. It is called the no-pair Hamiltonian of one electron in the
  quantum dot.
\item For scalar type magnetic operators, like Schr\"odinger operators
  $(\gp-\gA)^2 +V$ or Chandrasekhar operators $|\gp-\gA|+V$, it is
  known that $\gA$ does not lower the ground state energy because of
  the diamagnetic inequality. For operators involving spin in an
  essential way like the Pauli operator this is known to be
  false. Although, in our case, we cannot expect the energy to
  increase when $\gA$ is turned on, our result shows, that the energy
  is not lowered dramatically, i.e., the critical coupling constant is
  not lowered. Thus, the boundedness result can be interpreted as a
  weak form of the diamagnetic inequality.
\item The result for $Z>Z_c$ means physically that the electron is
  pulled into the nucleus of the interstitial atom as the trial
  function of the proof will indicate.
\item The critical coupling constant in the three dimensional
  non-magnetic case with arbitrary non-negative mass was found by
  Evans et al \cite{Evansetal1996}. Tix \cite{Tix1997,Tix1998}
  sharpened the result to strict positivity with a lower bound linear
  in the mass.
\item \label{nummer} The critical coupling constant in the
  2-dimensional non-magnetic case was investigated by Bouzouina
  \cite{Bouzouina2002}. An error in the constant he obtained was
  corrected by Walter \cite{Walter2010}.
\item The 3-dimensional magnetic case -- for a rather big class of
  magnetic fields -- was treated by Matte and Stockmeyer
  \cite{MatteStockmeyer2009}. They showed that the critical constant
  is not lowered by an intricate resolvent method.  The generality of
  their result is paid for by the absence of an explicit lower bound
  on the energy. The bonus of our direct approach based on Lieb and
  Yau's \cite{LiebYau1988} strategy in the variant found in
  \cite{Evansetal1996} -- compared to transfering the methods of
  \cite{MatteStockmeyer2009} -- is our result on the positivity of the energy.
\item The numerical value of the critical coupling constant is $Z_c
  \approx 0.3780$ which is compared with the expected critical
  coupling constant $\tilde Z_c$ of the existence of a distinguished
  self-adjoint extension of the non-magnetic Weyl operator
  $W_Z$. Pereira et al \cite{Pereiraetal2007} and Shytov et al
  \cite{Shytovetal2007} suggest in physical language and using
  physical arguments that $\tilde Z_c=1/2$. Recently Warmt
  \cite[Satz 2.2.6]{Warmt2011} showed that this is indeed the case.
\end{enumerate}

\section{The positive spectral subspace and partial wave
 analysis \label{sec:posspecpartwave}}

The fact, that we are dealing with spinors in the positive spectral
subspace of $W_0$ allows us to reduce the problem to unrestricted
scalar wave functions (see \cite[Section 1.1]{Franketal2009} for the
three dimensional case).
\begin{lemma}
 \label{lemma3}
 The map
 \begin{equation}
   \label{eq:14}
   \begin{split}
     \Phi: L^2(\rz^2) & \to \gH\\ 
     u & \mapsto \frac1{\sqrt{2}} \begin{pmatrix}
       u\\
       d^*|d^*|^{-1}u
     \end{pmatrix}
   \end{split}
 \end{equation}
is unitary.

Furthermore, its restriction to $\gq_0$ is a unitary map from $\gq_0$ to
$\gQ_0$ with the associated scalar products.
\end{lemma}
\begin{proof}
  First, we remark that $\Phi$ maps indeed to $\gH$. This holds, since
  $\gH$ is the closure of $\gQ_0$ in the $L^2$-norm.

  To show that $\Phi$ is surjective, assume $\psi=(u,v)^t \in \gH$ and
  orthogonal to
  $$\{(w,d^*|d^*|^{-1}w)^t|w\in L^2(\rz^2)\}.$$
  This implies
  $$ (u,w) + (|d^*|^{-1} dv,w) = ( u + |d^*|^{-1} dv, w) = 0$$
  for all $w\in L^2(\rz^2)$, i.e.,  $u= -|d^*|^{-1} dv$.

  Next we remark that
  $$
  \begin{pmatrix}
    -|d^*|^{-1} d\f mn\\
    \f mn
  \end{pmatrix}
  $$
  $n\in\nz_0$, $m\in\gz$ are eigenvectors with negative eigenvalue,
  namely $-2\sqrt {n+m_++1}$. Thus
  $$  \left(
    \begin{pmatrix}
      -|d^*|^{-1}d v\\v
    \end{pmatrix}
    , \begin{pmatrix}
      \f mn\\
      -d^*|d^*|^{-1}\f mn
    \end{pmatrix}\right) =0$$
  for all $n$ and $m$ which implies $(|d^*|^{-1}d v, \f mn)=0$,
  i.e., $dv=0$. Therefore, $\psi=(0,v)^t$. Such vectors are in the
  kernel of $W_0$, i.e., orthogonal to the positive spectral space, so
  that in the end $\psi=0$ is the only vector in the positive spectral
  space which is orthogonal to $\Phi(L^2(\rz^2,\cz^2))$.

  The identity $(u,v)_{L^2(\rz^2)}= (\Phi u, \Phi v)_\gH$ for all
  $u,v\in L^2(\rz)$ is immediate, as is the unitarity of the
  restriction.
\end{proof}
Using Lemma \ref{lemma3} we define the operator $w_Z:= \Phi^* W_Z\Phi$ on
$\gq_0$. The associated quadratic form on $\gq_0$ is
\begin{equation}
 \label{eq:18}
 (u,w_Z u) := (u,\Phi^* W_Z \Phi u)=(u,|d^*|u)-Z (u,Vu)
\end{equation}
with 
\begin{equation}
  \label{eq:18a}
  V=\frac12\left(\frac1{|\cdot|}+|d^*|^{-1}d \frac1{|\cdot|} d^*
    |d^*|^{-1}\right).
\end{equation}
\begin{corollary}
  \label{cor1}
  The operators $\Lambda W_Z$ on $\gQ_0$ and $w_z$ on $\gq_0$ are unitarily
  equivalent by Lemma \ref{lemma3}. In particular both operators and
  also the associated forms have all the same lower maximal bound.
\end{corollary}

Next we calculate the matrix elements of $w_Z$ in the orthonormal
basis given by the eigenfunctions $\f mn$ of $w_0$. First of all,
we remark that this matrix is diagonal in the angular momentum quantum
number $m$. We get for the matrix $t^m$ associated with $w_0$ the
following matrix elements
\begin{equation}
  \label{eq:t}
  t^m_{n,n'} := (\f mn, w_0 \f m{n'}) \delta_{n,n'}= 2\sqrt{n+m_++1} \delta_{n,n'}
\end{equation}
which is immediate from the eigenvalues equation \eqref{eigen}; for
the first summand of the potential $V$ (see \eqref{eq:18a}) we get the
matrix $v^{m,0}$ with matrix elements
\begin{equation}
  \label{eq:v0}
  \begin{split}
    v^{m,0}_{n,n'}=&  (\f mn, {1\over|\cdot|} \f m{n'})\\
    =&\frac1{\pi\sqrt{(n+1)_{|m|} (n'+1)_{|m|}}}
    \sum_{k=0}^{\min\{n,n'\}}{(k+1)_{|m|-\eh}\over(n-k+\eh)_\eh(n'-k+\eh)_\eh}
  \end{split}
\end{equation}
which is obtained by explicit calculation using the generating
function of the generalized Laguerre polynomials \cite[Formula
22.9.15]{Hochstrasser1965} and their recursion relations \cite[Formula
6.1.15]{Hochstrasser1965}.  (For convenience we use
Pochhammer's notation $(z)_a:=\Gamma(z+a)/\Gamma(z)$ [see also
\eqref{eq:poch}].)  Eventually, the second summand of
the potential $V$ yields the matrix $v^{m,1}$ with matrix elements
\begin{equation}
  \label{eq:v1}
  v^{m,1}_{n,n'}=  (d^*|d^*|^{-1}\f mn, {1\over|\cdot|}d^*|d^*|^{-1} \f m{n'})
  =
  \begin{cases}
    v^{m+1,0}_{n,n'} & \text{if } m\geq 0\\
    v^{m+1,0}_{n+1,n'+1} &\text{if } m<0.\\
  \end{cases}
\end{equation}
This can be obtained from \eqref{eq:v0} by observing that
\begin{equation}
 \label{eq:36}
 d^* \vert d^* \vert^{-1} \f mn= 
 \begin{cases} 
   \ \ \ri\f{m+1}n & \text{if } m  \ge 0,\\ 
   -\ri \f{m+1}{n+1}  & \text{if } m < 0. 
 \end{cases}
\end{equation}

Thus, the quadratic form $\cE_m$ of the matrix $(\f mn,w_Z\f m{n'})$,
for fixed angular momentum $m\in\gz$, on $l^2_0(\nz_0)$ -- the
subscript denotes sequences of compact support -- is given as
\begin{equation}
  \label{eq:9}
  \cE_m[a] = \sum_{n,n'=0}^\infty \overline{a_n}
  \left[ t^m_{n,n'} - \tfrac Z2 (v_{n,n'}^{m,0}+v_{n,n'}^{m,1})\right]a_{n'}.
 \end{equation}

 As mentioned in Appendix \ref{a1}, $(\f mn, w_Z\f{m'}{n'})= (\f mn,
 w_Z\f m{n'})\delta_{m,m'}$, i.e., both, potential and kinetic energy,
 are diagonal in $m$.  Thus,
\begin{equation}
 \label{eq:separariert}
 (u,w_zu) = \sum_{m\in\gz} \cE_m[a^m]
\end{equation}
where we write $a_n^m:=(\f mn,u)$ for the generalized Fourier
coefficients for $u\in\gq_0$ and where we collect those coefficients
with the same angular momentum quantum number $m$ and write
\begin{equation}
 \label{eq:27}
  a^m= (a_0^m,a_1^m,\ldots).
\end{equation}
Obviously, $(a_n^m)_{n\in\nz_0,\ m\in\gz}\in l_0^2(\nz\times\gz)$. 
\begin{lemma}
  \label{pre}
  The following facts for the matrix elements $v^{m,0}_{n,n'}$ of the
  Coulomb potential $1/|z|$ hold:
  \begin{itemize}
  \item For all $m\in \gz$ and $n,n'\in\nz_0$
    \begin{equation}
      \label{eq:pre1}
      0\leq v^{m,0}_{n,n'} =v^{|m|,0}_{n,n'}.
    \end{equation}
  \item For $m,n,n'\in \nz_0$
    \begin{align}
      \label{eq:pre2}
      v^{m,0}_{n,n'} &\geq v^{m+1,0}_{n,n'}.  
    \end{align}
  \end{itemize}
\end{lemma}
\begin{proof}
  The first claim -- including the remarkable positivity of all matrix
  elements -- is immediate from the explicit expression
  \eqref{eq:v0}.

 The second claim, i.e., monotony of the matrix elements in $m$,
 follows again from \eqref{eq:v0}, if
 $$ {\Gamma(k+m+3/2)\over \sqrt{\Gamma(n+m+2)\Gamma(n'+m+1)}}
 \leq  {\Gamma(k+m+1/2)\over \sqrt{\Gamma(n+m+1)\Gamma(n'+m+1)}}
$$
for $k\leq n,n'$. This is immediate from the functional
equation of the Gamma function.
\end{proof}
\begin{lemma}
 \label{lemma2}
 We have
 \begin{equation}
   \label{eq:11}
   0\leq(\f mn,V\f m{n'})\leq (\f 0n,V \f0{n'})
 \end{equation}
 for $n,n' \in \nz_0$ and $m \in \gz$.
\end{lemma}
\begin{proof}
 By \eqref{eq:pre2}
 $$ v^{m,0}_{n,n'} + v^{m,1}_{n,n'}\leq v^{0,0}_{n,n'} + v^{0,1}_{n,n'}$$
for all $m\geq0$. 

For negative $m$ 
$$
v^{m,0}_{n,n'}+v^{m,1}_{n,n'}\leq v^{-1,0}_{n,n'} + v^{-1,1}_{n,n'}
$$
where we use all three claims of Lemma \ref{pre}.
Thus, it suffices to show that
\begin{equation}
  \label{eq:lem2a}
   v^{-1,0}_{n,n'} + v^{-1,1}_{n,n'} \leq  v^{0,0}_{n,n'} + v^{0,1}_{n,n'}
\end{equation}
By \eqref{eq:v1} we have $v^{0,1}_{n,n'} = v^{1,0}_{n,n'}$ and
$v^{-1,1}_{n,n'} = v^{0,0}_{n+1,n'+1}$. Thus \eqref{eq:lem2a} is equivalent to
\begin{equation}
  \label{eq:lem2b}
  v^{0,0}_{n+1,n'+1}\leq v^{0,0}_{n,n'}.
\end{equation}
This is shown by induction, first in $n'$ and then in $n$.
\end{proof}
\begin{corollary}
  \label{cor2}
  For all $a\in l^2_0(\nz_0)$ we have
  $$\inf\{\cE_m[a] | a \in l^2_0(\nz_0)\} 
  \geq \inf\{ \cE_0[a] | a\in l^2_0(\nz_0) \}$$ and
  $$\inf\{(u,w_zu)| u\in \gq_0\}= \inf \cE_0[l^2_0(\nz_0)]. $$
\end{corollary}
\begin{proof}
  Since the kinetic energy $\sum_n t^m_{n,n} |a_n|$ is obviously
  invariant under the substitution $a \rightarrow |a|$ and since the
  potential energy
  $$ - Z \sum_{n,n'} \overline{a_n} {v^{m,0}_{n,n'}+v^{m,1}_{n,n'}\over2} a_{n'}$$
  decreases by the same substitution because of the positivity of the
  potential matrix elements (Lemma \ref{pre}, Formula \eqref{eq:pre1}),
  it suffices to take the infimum over positive sequences $a\in
  l^2_0(\nz_0)$ only. Thus, the desired inequalities follow from the
  corresponding inequalities of the matrix elements \eqref{eq:11}.
\end{proof}

\section{Proof of the theorem \label{sec:proof}}

\begin{proof} (Theorem \ref{theorem1}) By Corollary \ref{cor1} it
  is enough to study $w_Z$. By Corollary \ref{cor2}, this is
  equivalent to show lower boundedness of the quadratic form $\cE_0$
  on non-negative sequences $a\in l^2(\nz_0)$.

  At this point we embark on a strategy which goes back to Abel --
  at least -- and which has been introduced in relativistic quantum
  mechanics by Lieb and Yau \cite{LiebYau1988}); it basically consists
  of estimating a non-diagonal operator by a diagonal one using the
  Schwarz inequality suitably. We will apply it to the two potential
  matrices $v^0$ and $v^1$ with matrix elements $v^{0,0}_{n,n'}$ and
  $v^{0,1}_{n,n'}$ (\ref{eq:9}). (For the matrix elements we will,
  from now on, suppress the reference to $m=0$ as well and write
  simply $v^\sigma_{n,n'}$, $\sigma\in\{0,1\}$.)  Given any sequence
  $(g_n)_{n\in\nz_0}$ with positive entries this strategy suggests
  estimating as follows:
  \begin{equation}
    \label{eq:12}
    (a,v^\sigma a)  
    \le \sum_{n=0}^\infty {a_n^2 \over g_n} \sum_{n'=0}^\infty \v\sigma n{n'}g_{n'},
  \end{equation}
  where we use matrix notation on the left and that $v^\sigma$ is
  symmetric. (Note that we suppress an index $\sigma$ with $g$
  although $g$ can -- and will -- depend on $\sigma$.)

  We start with the case $\sigma=0$ and obtain 
  \begin{equation}
    \label{eq:13}
    \begin{split}
      (a, v^0 a) = & \frac1{\pi}\sum_{n,n'=0}^\infty a_na_{n'}
      \sum_{k=0}^{\min\{n,n'\}}{(k+1)_{-\eh}\over(n-k+\eh)_\eh(n'-k+\eh)_\eh}\\
      & \leq \frac1{\pi} \sum_{n=0}^\infty \frac{a_n^2}{g_n}
      \sum_{k=0}^n {(k+1)_{-\eh}\over(n-k+\eh)_\eh} \sum_{n'=0}^\infty
      \frac1{(n'+\eh)_\eh} g_{n'+k}
    \end{split}
  \end{equation}
  using \eqref{eq:12} and substituting $n'\rightarrow n'+k$.  We pick
  for $\sigma=0$
  \begin{equation}
    \label{g}
    g_n= \frac1{(n+\frac14)_{\frac34}}.
  \end{equation}
  This allows to explicitly do the summation in $n'$ and $k$ which gives
  \begin{equation}
    \label{eq:16}
    (a, v^0 a)  \leq {\Gamma(\frac14)^4\over2\pi^2}
    \sum_{n=0}^\infty  {\Gamma(n+\tfrac34)\over\Gamma(n+\tfrac14)} a_n^2.
  \end{equation}
  We apply Gautschi's inequality \eqref{eq:15} for $n\in\nz$ and we get 
  \begin{equation}
    \label{eq:o0}
    {\Gamma(n+\tfrac34)\over\Gamma(n+\tfrac14)}\leq \sqrt{n+\tfrac34}<\sqrt{n+1}
   \end{equation}
   which is also true for $n=0$ by inspection. Thus,
   \begin{equation}
     \label{eq:o00}
     (a,v^0 a) \leq  {\Gamma(\frac14)^4\over2\pi^2} \sum_{n=0}^\infty\sqrt{n+1} a_n^2.
   \end{equation}

   It remains to treat the case $\sigma=1$. We use again (\ref{eq:12})
   and obtain
  \begin{equation}
    \label{eq:131}
    \begin{split}
      (a, v^1 a) = & \frac1{\pi}\sum_{n,n'=0}^\infty {a_na_{n'}\over\sqrt{(n+1)(n'+1)}}
      \sum_{k=0}^{\min\{n,n'\}}{(k+1)_\eh\over(n-k+\eh)_\eh(n'-k+\eh)_\eh}\\
      \leq &\frac1{\pi} \sum_{n=0}^\infty \frac{a_n^2}{g_n\sqrt{n+1}}
      \sum_{k=0}^n {(k+1)_\eh\over(n-k+\eh)_\eh} \sum_{n'=0}^\infty
      { g_{n'+k}\over\sqrt{n'+k+1}(n'+\eh)_\eh}
    \end{split}
  \end{equation}
  substituting $n'\rightarrow n'+k$. In this case we pick
  \begin{equation}
    \label{eq:g1}
    g_n:= {\sqrt{n+1}\over (n+\tfrac34)_{\frac54}}
  \end{equation}
  which again allows for explicit summation in $n'$ and $k$ yielding
  \begin{equation}
    \label{eq:161}
    (a,v^1a)\leq{32\pi^2\over\Gamma(\frac14)^4}\sum_{n=0}^\infty {\Gamma(n+\tfrac54)\over\Gamma(n+\tfrac34)} a_n^2.
  \end{equation}

  For $n\geq2$ we have 
  \begin{equation}
    (n-\tfrac14)^{-1/2} < {\sqrt{n+1}\over n+\tfrac14}
  \end{equation}
  By Gautschi's inequality \eqref{eq:15} the left hand side majorizes
  $\Gamma(n+\tfrac14)/\Gamma(n+\tfrac34)$. Thus, by the Gamma function's
  functional equation we get
  \begin{equation}
    \label{eq:1o}
    {\Gamma(n+\tfrac54)\over\Gamma(n+\tfrac34)}< \sqrt{n+1}.
  \end{equation}
  However, this inequality is also true for $n=0$ and $n=1$ by inspection. Thus,
  \begin{equation}
    \label{eq:011}
    (a,v^1a)< {32\pi^2\over\Gamma(\frac14)^4}\sum_{n=0}^\infty \sqrt{n+1} a_n^2.
  \end{equation}

  Putting all together we have
  \begin{equation}
    \label{ab}
    \cE_0[a]\geq \sum_{n=0}^\infty 2 \sqrt{n+1}(1- Z/Z_c) a_n^2 \geq 0
  \end{equation}
  for $Z\leq Z_c$. Note that the first inequality in \eqref{ab} is
  indeed strict unless $a=0$ because of \eqref{eq:1o}. This shows the
  positivity of the form and therefore the first part of the theorem.

  This shows that
  \begin{equation}
    \label{qq}
    (u,v)_\gq := \sum_{n\in\nz_0,\ m\in\gz} \left((u,w_zv)
    +(u,v)\right)
  \end{equation}
  is a scalar product on $\gq_0$ and $(\psi,\varphi)_{\gQ_0} := (\Phi
  u,\Phi v)_\gQ$ is a scalar product on $\gQ_0$. The completions which
  we denote by $\gq$ and $\gQ$ are subspaces of $L^2(\rz^2)$ and $\gH$
  respectively. The quadratic form $(\psi,W_Z\psi)$ naturally extend to
  $\gQ$ and yields the self-adjoint Hamiltonian $B_Z$.

  For completenes we note that for $Z<Z_c$ Equation $\eqref{ab}$ shows
  that the norm $\|\cdot\|_\gq$ is equivalent to the ``Sobolev'' type
  norm $\|u\|_{W_0} := \sum_{m,n} (\sqrt{n+m_+} +1) |(\f mn,u)|^2$.

  To prove the claimed unboundedness we pick a family of trial
  sequences $a$ depending on an integer $N\in \nz$ -- for readability
  we refrain from indicating this explicitly -- given by
  \begin{equation}
    \label{eq:versuch}
    a_n:=
    \begin{cases}
      (n+1)^{-3/4} & \text{ if } n\leq N\\
      0 & \text{ if } n>N.
    \end{cases}
  \end{equation}
  We compute the expectation of the two summands $v^\sigma$,
  $\sigma\in\{0,1\}$, of the potential energy and obtain
  \begin{equation}
    \label{u}
    \begin{split}
    (a,v^\sigma a)= &\frac1\pi \sum_{k=0}^N (k+\sigma+1)_{-\eh}
    \left(\sum_{n=1}^{N-k+1}{(n)_{-\eh}\over (n+k)^{\frac34+\frac\sigma2}}\right)^2\\
    = &\frac1\pi\sum_{k=1}^N (k+\sigma+1)_{-\eh}\left(\int_0^\infty{\rd
        n\over
        (n+k)^{\frac34+\frac\sigma2}n^{\frac12}}\right)^2 +O(N^0) \\
    = &(\tfrac34+\tfrac\sigma2)_{-\eh}^2 \log(N) + O(N^0)  
    \end{split}
    \end{equation}
    for large  $N$. Thus
    \begin{equation}
      \cE_0[a] = (a,ta) -  Z(a,{v^0 +v^1\over2}a) 
      = 2(1- Z/Z_c) \log(N) + O(N^0),
      \end{equation}
      i.e., the form is unbounded from below for $Z>Z_c$.
\end{proof}

  \emph{Acknowledgment:} We acknowledge an enlightening discussion
  with Rupert Frank and Simone Warzel. We thank Heinrich K\"uttler for
  directing our attention to reference \cite{Gautschi1959}. We also
  thank Reinhod Egger und Alessandro De Martino for critical reading of
  the manuscript. Partial support by the DFG through the SFB-TR 12 is
  gratefully acknowledged.

\appendix

\section{Useful Facts on the Weyl Operator with Homogeneous Magnetic
  Field\label{a1}}

For the convenience of the reader and for fixing the notation we
collect in this appendix some facts related to the Weyl operator $W_0$
with homogeneous magnetic field.

We write $L_n^\alpha(x)$ for the $n$-th generalized Laguerre
polynomial with parameter $\alpha$ (Hochstrasser \cite[Formula
22.2.12]{Hochstrasser1965}). For $m\in\gz$ and $n\in\nz_0$ this allows
to define the functions
\begin{equation}
 \label{eq:6}
 \f mn(z)= \sqrt{\frac{n!}{\pi (n+\vert m \vert)!}}
 \begin{cases}
   e^{-\frac12 z \overline{z}} z^{\vert m \vert} L^{\vert m \vert}_n
   (z \overline{z})
   & \text{if } m \ge 0\\
   e^{-\frac12 z \overline{z}} \overline{z}^{\vert m \vert} L_n^{\vert
     m \vert} (z \overline{z}) & \text{if } m < 0.
 \end{cases}
\end{equation}
In polar coordinates $z=r\exp(\ri \varphi)$ these functions are written as
\begin{equation}
  \label{eq:a2}
  \f mn(r,\varphi)= \sqrt{\frac{n!}{\pi (n+\vert m \vert)!}} 
  e^{-\frac12r^2} r^{\vert m \vert} L^{\vert m \vert}_n (r^2)e^{\ri m\varphi}
\end{equation}
where -- in abuse of notation -- we use the same notation despite the
change of coordinates. Note that these functions form an orthonormal
basis of $L^2(\rz^2)$ which follows from the fact the
$(2\pi)^{-1/2}\exp(\ri m \varphi)$ are an orthonormal basis of
$L^2(0,2\pi)$ and for every fixed $m\in\nz_0$ the generalized Laguerre
polynomials $L^m_n$, $n\in\nz_0$, under suitable renormalization, are
an orthonormal basis of $L^2((0,\infty), r^me^{-r} \rd r)$ (Hewitt
\cite{Hewitt1954}). 

Using the recursion relations \cite[22.7.29-32]{Hochstrasser1965} of
the generalized Laguerre polynomials $L^m_n$ and
${L^m_{n}}'(x)=-L^{m+1}_{n-1}$ which is immediate from the
definition, we have
\begin{align}
  \label{eigen1}
  d^*\f mn&= 2 \ri \sgn(m) \sqrt{n+m_++1} \f{m+1}{n+\theta(-m)}\\
  \label{eigen2}d\f mn& =  -2\ri \sgn(m-1)\sqrt{n+m_+}\f {m-1}{n-\theta(-(m-1))}\\
  \label{eigen3}dd^* \f mn &= 4(n+m_++1) \f mn\\
  \label{eigen4}d^*|d^*|^{-1} \f mn & = \ri  \sgn(m) \f{m+1}{n+\theta(-m)}
\end{align}
for $ n\in \nz,\ m\in\gz$,  where -- as usual -- $m_+:=\max\{0,m\}$,
$$\theta(x):=
\begin{cases}
  1 &\text{if } x>0\\
  0 &\text{if }x\leq 0
\end{cases}
\text{and } \sgn(x):=
\begin{cases}
  1 &\text{if } x\geq0\\
  -1 &\text{if }x<0
\end{cases}
$$
(see also \cite[Section 7.1.3]{Thaller1992}). Note that this solution
is related to the non-relativistic Schr\"odinger equation with
homogeneous magnetic field in two dimensions (Fock
\cite{Fock1928}).

The angular momentum operator $L$ is given as
\begin{equation}
  \label{eq:drehimp}
  L:= x_1p_2-x_2p_1 = z\partial-\bar z\bar\partial = \frac 1\ri \partial_\varphi.
\end{equation}
Writing the $\f mn$ in spherical coordinates easily shows that they
are eigenfunctions of $L$ with eigenvalue $m$, i.e.,
\begin{equation}
  \label{eq:dreheigen}
  L\f mn = m \f mn.
\end{equation}

Since the $\f mn$ form an orthonormal basis, the
eigenvalue equation \eqref{eigen} implies that $|d^*|$ is invertible and
$d^*|d^*|^{-1}$ is an isometric operator. The spinors
\begin{equation}
 \label{eq:5}
 \p mn = \frac1{\sqrt{2}}
 \begin{pmatrix}
   \f mn\\
   d^* |d^*|^{-1} \f mn
 \end{pmatrix}
\end{equation}
for $n\in\nz_0$ and $m\in\gz$ form an orthonormal basis of $\gH$ as shown in
the proof of Lemma \ref{lemma3}.

Using \eqref{eigen1} through \eqref{eigen4} we find
\begin{equation}\label{eigen}
W_0\p mn = 2\sqrt{n+m_++1} \p mn,\ n\in \nz,\ m\in\gz,
\end{equation}
i.e, for fixed $m\in\gz$, the spinor $\p mn$ is the $n$-th eigenvector of
$W_0$ on the positive spectral subspace $\gH$.

The total angular momentum operator $J$ on $L^2(\rz^2)$ is given as
\begin{equation}
  \label{eq:drehimpg}
  J = L + 
  \frac12\begin{pmatrix}
    1&0\\
    0&-1
  \end{pmatrix}.
\end{equation}
The Formulae \eqref{eigen4}, \eqref{eq:dreheigen}, and \eqref{eq:5} imply
\begin{equation}
  \label{dreheigeng}
  J\p mn = (m+1/2) \p mn.
\end{equation}
In fact, 
\begin{equation}
  \label{eq:dreh}
   [W_0,J]=0 \text{ and } [w_0,L]=0.
\end{equation}

Eventually note, that $ \Phi^* J \Phi = L +1/2$ where $\Phi$ is the
unitary map defined Lemma \ref{lemma3}. This is the reason why it is
equally natural to label the basis by the orbital angular momentum
quantum number $m$ as to label it by the total angular momentum
quantum. We choose $m$ since the formulae are easier to handle.

\section{Some Useful Facts on Gamma and Related Functions\label{a2}}

The Gamma function 
$$ \Gamma(z):=\int_0^\infty z^te^{-t}{\rd t\over t}$$
is obviously positive on the positive half axis $\rz_+$ where it is also
analytic and log-convex (see, e.g., Rudin \cite[Theorem 8.18]{Rudin1976}).

A useful combination of Gamma functions is the
Pochhammer symbol
\begin{equation}
  \label{eq:poch}
  (z)_a:= {\Gamma(z+a)\over \Gamma(z)}
\end{equation}
which is a meromorphic function in both variables $z$ and $a$.

\begin{lemma}[Gautschi {\cite[Formulae 6 and 7]{Gautschi1959}}]
  \label{gautschi}
  For $x\in\rz_+$ and $0\le s\le 1$
  \begin{equation}
    \label{eq:15}
    (x+1)^{s-1}\leq \frac{\Gamma(x+s)}{\Gamma(x+1)} < x^{s-1}.
  \end{equation}
  \end{lemma}
Note that Gautschi claims the inequalities for $x\in \nz$
only. However, his proof is valid also for $x\in\rz_+$.

Furthermore, we note the reflection formula {\cite[Formula
  6.1.17]{Davis1965}} which states that for $0<\Re z<1$
  \begin{equation}
    \label{eq:refl}
    \Gamma(z)\Gamma(1-z)= {\pi\over \sin(\pi z)}.
  \end{equation}


\end{document}